\documentclass[letterpaper, 10 pt, conference]{ieeeconf}

\IEEEoverridecommandlockouts                              
\usepackage{cite}

\usepackage{ifpdf}
\ifpdf
\usepackage{graphicx}
\usepackage{epstopdf}
\else
\usepackage{graphicx}

\fi

\DeclareFontEncoding{LS1}{}{}
\DeclareFontSubstitution{LS1}{stix}{m}{n}
\DeclareSymbolFont{stixletters}{LS1}{stix}{m}{it}
\DeclareMathAccent{\cev}{\mathord}{stixletters}{"91}
\DeclareMathAccent{\vec}{\mathord}{stixletters}{"92}
\DeclareMathAccent{\vecev}{\mathord}{stixletters}{"95}
\usepackage{mathptmx} 
\usepackage{times} 
\usepackage{amsmath} 
\usepackage{amssymb}  
\usepackage{tabularx}
\usepackage{algorithm}
\usepackage{algorithmic}
\usepackage{ulem}
\usepackage{dsfont}

\usepackage[dvipsnames,usenames]{color}
\usepackage[colorlinks,%
linkcolor=BrickRed,%
filecolor=BrickRed,%
citecolor=RoyalPurple,%
]{hyperref}

\usepackage{tikz}

\newcommand{\cevF}{{\cev {\mathcal F}}}

\usepackage{tabularx}
\newcommand{\lr}[2]{\langle #1, #2 \rangle}

\newcommand{\newP}[1]{\noindent{\bf #1:}}

\newcommand{\ud}{\mathrm{d}}
\def\Re{\mathbb{R}}

\def\argmin{\mathop{\text{\rm arg\,min}}}

\newtheorem{theorem}{Theorem}

\newtheorem{lemma}{Lemma}
\newtheorem{remark}{Remark}
\newtheorem{proposition}{Proposition}

\newcommand{\trace}{\text{Tr}}

\newcounter{rmnum}

\newcounter{anum}

\title{\LARGE \bf
	Time-Reversal of Stochastic Maximum  Principle}

\author{Amirhossein Taghvaei
	\thanks{A. Taghvaei  is with the William E. Boeing Department of Aeronautics and Astronautics at University of Washington, Seattle {\tt\small amirtag@uw.edu}.}
}

\begin{document}
	\maketitle
	\thispagestyle{empty}
	\pagestyle{empty}
	
	\begin{abstract}
		Stochastic maximum   principle (SMP)  specifies a necessary   condition for the solution of a stochastic optimal control problem. The condition involves a coupled system of forward and backward stochastic differential equations (FBSDE) for the state and the adjoint processes. Numerical solution of the FBSDE  is challenging because the boundary condition of the adjoint process  is specified at the terminal time, while the solution should be adaptable to the forward in time filtration of a  Wiener process. In this paper, a ``time-reversal'' of the FBSDE system is proposed that involves integration with respect to a backward in time Wiener process. The  time-reversal is used to propose an iterative Monte-Carlo procedure to solves the FBSDE system and its time-reversal simultaneously. 
		The procedure involves approximating the {F\"ollmer's drift} and solving a regression problem between the state and its adjoint at each time. The procedure  is illustrated for the linear quadratic (LQ)  optimal control problem with a numerical example. 
	\end{abstract}

	\section{Introduction}
	\label{sec:intro}
	This paper is  concerned with the continuous-time stochastic optimal control problems that involve the 
	controlled stochastic differential equation  (SDE)
	
	\begin{equation}\label{eq:dyn}
		\ud X_t = a(X_t,U_t) \ud t + \sigma(X_t) \ud W_t,\quad X_0 \sim p_0, 
	\end{equation}
	where $X:=(X_t)_{t\geq 0}$  is the $n$-dimensional state trajectory, $ U:=( U_t)_{t\geq 0}$  is the $m$-dimensional control input,  $ W:= (W_t)_{t\geq 0}$  is the standard $n$-dimensional Wiener process,  $p_0$  is a probability distribution on $\mathbb R^n$, and $a:\Re^n \times \Re^m \to \Re^n$ and  $\sigma:\Re^n  \to \Re^{n\times n}$ are smooth and globally Lipschitz functions\footnote{In this article, the model is restricted to the case that the diffusion term  is not affected by the control input.}. The control input $U_t$ is assumed to be adapted to the filtration $\mathcal F_t=\sigma\{W_s;0\leq s\leq t\}$, generated by the Wiener process. 
	The objective of the optimal control problem is
	\begin{equation}\label{eq:control-cost}
		\min_{U_t\in \mathcal F_t}\,J(U):=\mathbb E\left[\int_0^{T}\ell(X_t,U_t) \ud t + \ell_f(X_{T})\right],
	\end{equation} 
	where $\ell:\Re^n \times \Re^m \to \Re$ and $\ell_f:\Re^n \to \Re$ are smooth functions, and the notation $U_t \in \mathcal F_t$ means $U_t$ is $\mathcal F_t$-measurable.

	The stochastic maximum principle (SMP) provides a necessary condition for the optimal control input. In order to express SMP, define the Hamiltonian function
	\begin{subequations}
		\begin{equation}\label{eq:H}
			H(x,u,y,z) := \ell(x,u) + y^\top a(x,u) +  \trace(z^\top \sigma(x) ),
		\end{equation}
		for $(x,u,y,z)\in \mathbb R^n \times \Re^m \times \Re^n \times \Re^{n\times n}$
		and the Hamiltonian system of forward-backward SDE (FBSDE) 
		\begin{align}\label{eq:forward-sde}
			\ud X_t &= \frac{\partial H}{\partial y}(X_t,U_t,Y_t,Z_t)\ud t + \sigma(X_t) \ud {W}_t,~ X_0\sim p_0,\\
			-\ud Y_t &= \frac{\partial H}{\partial x}(X_t,U_t,Y_t,Z_t) \ud t- Z_t \ud {W}_t,~ Y_{T}= \frac{\partial \ell_f}{\partial x} (X_{T}).\label{eq:backward-sde}
		\end{align}
		The solution to the FBSDE is a $\mathcal F_t$-adapted  triple $(X,Y,Z)=(X_t,Y_t,Z_t)_{0\leq t\leq T}$. 
		SMP states that if $ \hat U$  is the optimal control input for the objective~\eqref{eq:control-cost} subject to the constraint~\eqref{eq:dyn}, then 
		\begin{align}\label{eq:min-condition}
			\hat U_t \in \argmin_{u }\,H( \hat X_t,u, \hat  Y_t, \hat  Z_t),\quad \text{for}\quad t \in[0,T], 
		\end{align}
		where $\hat  X$, $ \hat Y$, and $ \hat Z$ solve the FBSDE~\eqref{eq:forward-sde}-\eqref{eq:backward-sde} with $U=\hat U$~\cite[Thm. 3.2, Case 1]{yong1999stochastic}. 
	\end{subequations}

	Numerical implementation of the FBSDE  is challenging and it has been subject of research since its introduction. The challenges include (i) obtaining a $\mathcal F_t$-adapted solution $Y_t$ while satisfying a terminal condition, (ii) and obtaining the adjoint diffusion coefficient $Z_t$.  There are generally two approaches to numerically solve the FBSDE, the PDE-based approaches that use the fact that $Y_t = \phi(t,X_t)$, where $\phi$ solves a parabolic pde~\cite{peng1991probabilistic,ma1994solving,ma2002numerical}, and direct numerical approxiamtion of the FBSDE that involve approximating  conditional expectations with respect to the filtration~\cite{zhang2004numerical,bouchard2004discrete,zhao2006new,bender2007forward,zhang2013sparse}.   
	
	In this paper, we explore an alternative approach to numerically solve the FBSDE system that is based on time-reversal of diffusions~\cite{anderson1982reverse,haussmann1986time,millet1989integration,follmer2005entropy,follmer2006time,cattiaux2021time}.   
	The starting point is to formulate the optimal control problem for the time-reversal of the SDE~\eqref{eq:dyn} where the solution $X_t$ and the control $U_t$ become adaptable to a backward filtration of a Wiener process (see~\eqref{eq:dyn-backward}).  Compared to~SDE~\eqref{eq:dyn}, the time-reversal involves additional drift term, referred to as the F\"ollmer's drift~\cite{follmer2005entropy}, that depends on the probability density of $X_t$.  As a result, the reversed SDE is interpreted as a controlled McKean-Vlasov SDE leading to a mean-field type optimal control problem~\cite{bensoussan2013mean,fornasier2014mean,carmona2018probabilistic}. The appropriate SMP for a mean-field optimal control problem  leads to a ``time-reversal'' of the FBSDE~\eqref{eq:forward-sde}-\eqref{eq:backward-sde}, where the solution is required to be adapted to a backward filtration (see Sec.~\ref{sec:FBSDE-reversal}).  
	The time-reversal procedure is justified in Prop.~\ref{prop:FBSDE-reversed} where it is shown that the solution to~\eqref{eq:forward-sde}-\eqref{eq:backward-sde} and its time-reversal admit the same probability law.

The time-reversed FBSDE is numerically useful as it facilitates the numerical integration of the adjoint process. Based on this, an iterative numerical procedure is proposed in Sec.~\ref{sec:numerics} that simultaneously solves the FBSDE and its time-reversal. 
	Each iteration of this algorithm involves (i) the forward Monte-Carlo simulations of the state process; (ii) approximation of the F\"ollmer's drift; (iii) backward  simulation of the state and adjoint processes according to the time-reversed FBSDE;  and (iv) a gradient descent update of the control input to minimize the Hamiltonian function.  The third step   involves solving a regression problem between the state  and adjoint for all time $t\in[0,T]$. 
	  The proposed procedure is presented   for the linear quadratic (LQ) setting, where the approximation of the F\"ollmer's drift simplifies to mean and covariance estimation,  and the regression problem becomes linear. A two-dimensional numerical example is provided to illustrates the performance of the algorithm, while the extension to the nonlinear setup and extensive numerical experiment is subject of future work.

	
	\newP{Relation to the other approaches} This paper is closely related to the recent works that solve stochastic optimal control problem by approximating the solution to the FBSDE system~\cite{exarchos2016learning,han2018solving,exarchos2018stochastic,archibald2020stochastic,pereira2020feynman}. In particular, similar to~\cite{exarchos2016learning}, the proposed approach  involves solving a regression type problem at each time. However, the numerical approximation of the adjoint process is different in comparison to~\cite[Eq. 18]{exarchos2016learning} as it involves an additional term due to the time-reversal effect. 

	
	\section{Preliminaries}
	This section contains a brief introduction to the  technical concepts that are essential to the proposed procedure.  
	\subsection{Backward stochastic integration}
	We follow the notation  introduced in~\cite[Sec. 2]{pardoux1987two}. 
	We use $W_t$ to denote a standard $n$-dimensional Wiener process
	which introduces the forward and backward filtrations 
		\begin{align*}
			{\mathcal F}_t:=&\sigma(W_s;0\leq s\leq t),\quad 
			\cev {\mathcal F}_{t}:=\sigma(W_T-W_s;t\leq s\leq T),
		\end{align*}
	for $t\in[0,T]$ (the backward filtration is increasing as $t \downarrow 0$). 
	%
	For any continuous $\cevF_t$-measurable process $\cev V_t$, 
	the backward stochastic integral 
	\begin{align*}
		\int_{T-t}^T  \cev V_s   \cev \ud  W_s := -\int_0^{t} \cev V_{T-s}  \ud \tilde W_s
	\end{align*}
	where the right-hand-side is the usual  It\^o integral with respect to the Wiener process  $\tilde W_t:=W_{T-t}-W_{T}$.	Numerically, the backward stochastic integral  may be approximated as  	
	\begin{equation*}
		\int_{T-t}^T \cev V_s   \cev \ud  W_s  = \lim_{n \to \infty }\sum_{k=0}^{n-1} \cev V_{T-k\Delta t}(W_{T-k\Delta t} - W_{T-(k+1)\Delta t})
	\end{equation*}
	where $\Delta t = \frac{t}{n}$.
	
	We use the notation  $\cev \ud W_t$ to express stochastic differential equations (SDEs) that involve backward stochastic integration. 
	For example, a SDE of the form
	\begin{equation}
		\ud \cev X_t  = a(\cev X_t) \ud t + \sigma (\cev X_t) \cev \ud W_t,\quad \cev X_T = \xi,\label{eq:reverse-sde}
	\end{equation}
	is a short-hand notation for 
	\begin{align}
		\cev X_{T-t}  &= \xi - \int_{T-t}^T a(\cev X_s) \ud s - \int_{T-t}^T  \sigma (\cev X_s) \cev \ud W_s,~ \forall t \in [0,T]
		\label{eq:integral-form}
	\end{align}
	which defines a $\cevF_t$-adapted solution $\cev X_t$. 
	\begin{lemma}\label{lem:reverse}
		The SDE~\eqref{eq:reverse-sde} is equivalent to the SDE 
		\begin{equation}\label{eq:lem-forward-SDE}
			\ud \tilde X_t  = -a(\tilde X_t) \ud t + \sigma (\tilde X_t) \ud \tilde  W_t,\quad \tilde  X_0 = \xi,
		\end{equation}
		where  $\tilde W_t :=  W_{T-t} - W_T$  and $\tilde X_t := \cev X_{T-t}$ is adapted to the forward filtration $\tilde {\mathcal F}_t :=\sigma\{\tilde W_s;0\leq s\leq t\}$, for $t\in[0,T]$.
	\end{lemma}
	\begin{proof}
		Using $\tilde X_t = \cev X_{T-t}$ in~\eqref{eq:integral-form} yields 
		\begin{align*}
			\tilde X_t &= \xi - \int_{T-t}^T a(\tilde X_{T-s}) \ud s - \int_{T-t}^T  \sigma (\cev X_{s}) \cev \ud  W_s\\&=\xi - \int_0^t a(\tilde X_s) \ud s + \int_0^t  \sigma (\tilde X_{s}) \ud \tilde W_s,
		\end{align*}
		concluding the forward in time  SDE~\eqref{eq:lem-forward-SDE}. 
	\end{proof}

	\subsection{Time reversal of SDE} \label{sec:time-reversal}
	Consider the two stochastic processes $X = \{X_t;\,0\leq t\leq T\}$ and $\cev X = \{\cev X_t;\,0\leq t\leq T\}$ that are governed by the following SDEs
	\begin{subequations}
		\begin{align}
			\ud X_t& = \tilde a(t,X_t) \ud t + \sigma(X_t) \ud W_t,~ X_0 \sim  p_{0}, \label{eq:sde}\\
			\ud \cev X_t &= \tilde a(t,\cev X_t) \ud t + b(\cev X_t,p_{X_t})\ud t+\sigma(\cev X_t) \cev \ud { W}_t,~ \cev X_{T}  \sim p_{X_{T}} ,\label{eq:sde-reverse}
		\end{align} 
		where,  
		\begin{equation}\label{eq:b-def}
			b(x,p) :=- \frac{1}{p(x)}\nabla \cdot (D(x) p(x)),
		\end{equation}
		$D(x) := \sigma(x)\sigma(x)^\top$, 
		and $p_{X_t}$ is the probability density function of $X_t$.  The F\"ollmer's drift  term~\eqref{eq:b-def} is also interpreted as the score function when $D(x)=I$. 
	\end{subequations}
	The following result from~\cite[Thm. 1.12]{cattiaux2021time} relates the two processes. 
	\begin{theorem}\label{thm:sde-reversal}
		Assume $D(x) \succ 0$ for all $x \in \Re^n$ and the probability law of $X$ satisfies a certain finite entropy condition  (see~\cite[Thm. 1.12]{cattiaux2021time}). 
		Then the probability law of  $X $ and $\cev X$, as stochastic processes, are equal, i.e. $X \overset{d}{=} \cev X$.  
	\end{theorem}   
	\begin{remark}
		The standard statement  of the time reversal result involves writing the SDE for the time-reversed process $\tilde X_t := X_{T-t}$. Theorem~\ref{thm:sde-reversal} is an equivalent statement: application of Lemma~\ref{lem:reverse} to~\eqref{eq:sde-reverse} gives the SDE for $\cev X_{T-t} \overset{d}{=} X_{T-t} = \tilde X_t$. 
	\end{remark}
	\medskip
	
	In this paper, we like to apply the  result of Thm.~\ref{thm:sde-reversal} to the case of controlled SDEs of the form~\eqref{eq:dyn}. This is possible whenever the control input $U_t$ follows a deterministic state feedback control law of the form $U_t=k(t,X_t)$ for a time-varying function $k:[0,T]\times \Re^n \to \Re^m$.  Under this assumption, we can define $\tilde a(t,x):=a(x,k(t,x))$ and apply  Thm. 1. 
	This  assumption on the control is not restrictive because the optimal control law of the stochastic optimal control problem~\eqref{eq:control-cost} is known to be of this form.  

	\subsection{Wasserstein derivatives}
	Let $\mathcal P_{2}(\Re^n)$ be the space of  probability distributions with finite-second moments. For any $\mu \in \mathcal P_{2}(\Re^n)$, we consider the inner-product of two vector-fields on $\Re^n$ as 
	\begin{equation*}
		\lr{v_1}{v_2}_\mu = \int_{\Re^n} v_1(x)^\top v_2(x)\ud \mu(x),
	\end{equation*}
	which induces the norm $\|v\|_\mu:=\sqrt{\lr{v}{v}}_\mu$. 
	The tangent space $\text{Tan}_\mu( \mathcal P_{2}(\Re^n))$ is identified with the closure of the space of vector-fields of gradient form $\nabla \phi$ with respect to $\| \cdot\|_\mu$ norm~\cite[Lemma 8.4.2]{ambrosio2008gradient}.  The tangent space is equipped with the Riemannian metric $\lr{\nabla \phi_1}{\nabla \phi_2}_\mu$. The Wasserstein derivative of a  real-valued function $F:  \mathcal P_{2}(\Re^n) \to \Re$ at $\mu$, denoted by $\frac{\partial F}{\partial \mu}(\mu) \in \text{Tan}_\mu( \mathcal P_{2}(\Re^n))$, satisfies the identity 
	\begin{align*}
		\lr{\nabla F(\mu)}{\nabla \phi} = \frac{\ud}{\ud t}F((\text{Id}+t\nabla \phi)_\#\mu)\vert_{t=0},
	\end{align*}
	for all $\nabla \phi \in \text{Tan}_\mu( \mathcal P_{2}(\Re^n))$, where $\text{Id}$ denotes the identity map and $\#$ denotes the push-forward operator. See \cite[Def. 5.62]{carmona2018probabilistic} for a rigorous definition of the Wasserstein derivative. 
	In this article, we use the following property of the Wasserstein derivatives, shown in~\cite[Example 3 on page 387]{carmona2018probabilistic}: 
If $F(\mu) = \int_{\Re^n} f(x',\mu)\ud \mu(x')$ and $f$ is differentiable, then
	\begin{align}\label{eq:W-derivative}
		\frac{\partial F}{\partial \mu} (\mu)(x)= \frac{\partial f}{\partial x} (x,\mu) + \int \frac{\partial f}{\partial \mu}(x',\mu)(x)\ud \mu(x'). 
	\end{align}
	
	\subsection{PDE solution of FBSDE}
	Consider the FBSDE~\eqref{eq:forward-sde}-\eqref{eq:backward-sde} assuming the control $U_t = k(t,X_t,Y_t,Z_t)$ for some smooth function $k:[0,T]\times \Re^n \times \Re^n \times \Re^{n\times n} \to \Re^m$. 
	The PDE solution of the FBSDE is based on the hypothesis that  $Y_t = \phi(t,X_t)$ for some smooth function $\phi:[0,T]\times \Re^n \to \Re^n$. Then, by It\^o's formula
	\begin{align*}
		\ud Y_t = &	\frac{\partial \phi}{\partial t} (t,X_t)\ud t + \frac{\partial \phi}{\partial  x}(t,X_t)  \frac{\partial H}{\partial y}(X_t,k(t,X_t,Y_t,Z_t),Y_t,Z_t) \ud t\\&+ \frac{\partial \phi}{\partial  x}(t,X_t)  \sigma(X_t)\ud W_t  + \frac{1}{2}\trace(D(X_t) \frac{\partial^2 \phi}{\partial  x^2}(t,X_t) ) \ud t,
	\end{align*}
	where $\frac{\partial \phi}{\partial  x}$ and $\frac{\partial^2 \phi}{\partial  x^2}$ are the Jacobian and Hessian respectively. 
	Comparing  to~\eqref{eq:backward-sde} conlcudes that  $Z_t =  \frac{\partial \phi}{\partial  x}(t,X_t)  \sigma(X_t)$ and $\phi$ satisfies the PDE
	\begin{align}\nonumber 
		&\frac{\partial \phi}{\partial t}+ 	\frac{\partial \phi}{\partial x} \frac{\partial H}{\partial y} (t,x,k,\phi,\frac{\partial \phi}{\partial  x} \sigma)+ \frac{1}{2}\trace(D \frac{\partial^2 \phi}{\partial  x^2}) \\&+ \frac{\partial H}{\partial x} (t,x,k,\phi,\frac{\partial \phi}{\partial  x} \sigma)= 0,~ \phi(T,x) = \frac{\partial \ell_f}{\partial x} (x),\label{eq:FBSDE-PDE}
	\end{align}
	for all $(x,t)\in\Re^n\times[0,T]$, where, for convenience,  $k$ and $\sigma$ are used to denote $k(t,x,y,z)$ and $\sigma(x)$. The existence of the solution to this parabolic PDE is investigated in~\cite{ma1994solving} under certain regularity assumptions. In this paper, we make the assumption that the solution to the PDE~\eqref{eq:FBSDE-PDE} exists. 
	
	\section{Time-reversal procedure}
	In this section, we introduce the ``time-reversal" formulation of the stochastic control problem~\eqref{eq:dyn}-\eqref{eq:control-cost}. Based on the discussion presented in Sec.~\ref{sec:time-reversal}, we consider  the  time-reversal of~\eqref{eq:dyn} according to 
	\begin{align}
		\ud \cev X_t =& a(\cev X_t,\cev U_t) \ud t + b(\cev X_t,p_{\cev X_t})\ud t+\sigma(\cev X_t) \cev \ud { W}_t,\quad  \cev X_{T}  \sim p_T,\label{eq:dyn-backward}
	\end{align}
	where 
	$\cev \ud  W_t$ denotes a  backward stochastic integration, $b(x,p)$ is defined in~\eqref{eq:b-def}, and 
	$p_{\cev X_t}$ is the probability density function of $\cev X_t$. We consider control inputs $\cev U_t$ that are adapted to the backward filtration $\cev {\mathcal F}_t$ that lead to an $\cevF_t$-adapted solution $\cev X_t$. The time-reversal of the objective~\eqref{eq:control-cost}  becomes 
	\begin{equation}\label{eq:cost-backward}
		\min_{\cev U_t \in \cevF_t }\,\mathbb E\left[\int_0^{t_f} \ell(\cev X_t,\cev U_t) \ud t + \ell_f(\cev X_T)\right]. 
	\end{equation}

	The SDE~\eqref{eq:dyn-backward} involves an additional drift  term $b(\cev X_t,p_{\cev X_t})$ 
	that depends on the probability distribution of ${\cev X_t}$. As a result, it belongs to the class of McKean-Vlasov SDEs, which  requires an appropriate extension of the SMP~\cite[Sec. 6.3]{carmona2018probabilistic}.  
	
	\subsection{SMP for the time-reversed problem}\label{sec:FBSDE-reversal}
	In order to describe the maximum principle for the McKean-Vlasov SDE~\eqref{eq:dyn-backward} with the objective function~\eqref{eq:cost-backward},  define the (backward) Hamiltonian
	\begin{subequations}
		\begin{align}\label{eq:super-H}
			\cev H(x,p,u,y,z) := 
			\ell(x,u) + y^\top (a(x,u) 
			+ b(x,p))
			- \trace(z^\top \sigma(x)).
		\end{align}
	for $(x,p,u,y,z) \in \Re^n \times \mathcal P_2(\Re^n)\times \Re^m \times \Re^n \times \Re^{n\times n}$. Note that, in comparison to the Hamiltonian~\eqref{eq:H}, the sign of $z$ is changed to account for the backward stochastic  integration. 
		The backward Hamiltonian leads to
		the FBSDE 
		\begin{align}\label{eq:forward-sde-reversed}
			\ud \cev X_t &= \!\frac{\partial \cev H}{\partial y}(\cev X_t,p_{\cev X_t},\cev U_t,\cev Y_t,\cev Z_t)\ud t \!+\! \sigma(\cev X_t) \ud \cev{W}_t,~\cev X_{T}\!\sim\! p_T\\\label{eq:backward-sde-reversed}
			-\ud \cev Y_t &= \frac{\partial \cev   H}{\partial x}(\cev X_t,p_{\cev X_t},\cev U_t,\cev Y_t,\cev Z_t) \ud t-  \cev Z_t \ud \cev{W}_t\\&+  \tilde {\mathbb E}\left[\frac{\partial \cev  H}{\partial p}(\tilde X_t,p_{\cev X_t},\tilde U_t,\tilde Y_t,\tilde Z_t)(\cev X_t)\right]\ud t, \quad \cev Y_{t_f} = \frac{\partial \ell_f}{\partial x} (\cev X_{t_f}) \nonumber, 
		\end{align}
		where  $\frac{\partial\cev  H}{\partial p}(x,p,u,y,z):\mathbb R^n \to \mathbb R^n$  is the Wasserstein derivative of $\cev  H$ with respect to $p$, and  the expectation  $\mathbb {\tilde E}$ is  with respect to  $\{\tilde X_t, \tilde Y_t,\tilde Z_t,\tilde U_t\}$ which are independent copies of $\{\cev X_t,\cev Y_t,\cev Z_t,\cev U_t\}$. 
		The SMP states that if $\hat U$  is the optimal control input, then  
		\begin{align}\label{eq:min-condition-reversed}
			\hat U_t \in \argmin_{u \in \mathbb R^m}\,\cev  H( \hat X_t,  p_{\hat X_t},u, \hat  Y_t, \hat  Z_t) 
		\end{align}
		where $\hat  X$, $ \hat Y$, and $ \hat Z$ solve the FBSDE~\eqref{eq:forward-sde-reversed}-\eqref{eq:backward-sde-reversed} with $\cev U=\hat U$~\cite[Thm. 6.14]{carmona2018probabilistic}. 
	\end{subequations}
	
	In the next section, we show that the new  FBSDE~\eqref{eq:forward-sde-reversed}-\eqref{eq:backward-sde-reversed} and the original FBSDE~\eqref{eq:forward-sde}-\eqref{eq:backward-sde} both admit the same probability law for the state and adjoint processes. 
	
	\subsection{Time-reversal relationship}
	We present the result under the assumption that the control inputs $U_t$ and $\cev U_t$ follow the same feedback control law
	\begin{equation}\label{eq:U-feedback}
		U_t = k(t,X_t,Y_t,Z_t),\quad \cev U_t=  k(t,\cev X_t,\cev Y_t,\cev Z_t),
	\end{equation}
for some smooth function $k$.  
	\begin{proposition}\label{prop:FBSDE-reversed}
		Let $(\cev X,\cev Y,\cev Z)$ be the  triple that solves the FBSDE~\eqref{eq:backward-sde-reversed}-\eqref{eq:forward-sde-reversed} with $\cev U_t$ as in~\eqref{eq:U-feedback}.  Assume  the PDE~\eqref{eq:FBSDE-PDE} has a smooth solution $\phi$. Then, 
		\begin{align}\label{eq:FBSDE-reversed-1}
			(\cev Y_t,\cev Z_t)= ( \phi(t,\cev X_t),\frac{\partial \phi}{\partial x} (t,\cev X_t)\sigma(\cev X_t)),\quad \text{for}\quad  t\in [0,T].  
		\end{align} 
		Moreover, 
		assume $(X,Y,Z)$ solves the FBSDE~\eqref{eq:forward-sde}-\eqref{eq:backward-sde}  with $U_t$ as in~\eqref{eq:U-feedback}, and $p_T$ is equal to the probability density function of $X_T$. Then,  the probability law of the quadruple $(\cev X,\cev U,\cev Y,\cev Z)$ is equal to the probability law of  $(X,U,Y,Z)$, i.e. 
		\begin{equation}\label{eq:FBSDE-reversed-2}
			(X,U,Y,Z) \overset{d}{=}(\cev X,\cev U,\cev Y,\cev Z). 
		\end{equation}
	\end{proposition}
	
	\medskip 
	\begin{remark}
		This result establishes that the new FBSDE~\eqref{eq:backward-sde-reversed}-\eqref{eq:forward-sde-reversed}  is indeed the time-reversal of the FBSDE~\eqref{eq:forward-sde}-\eqref{eq:backward-sde}. 
		The key element that allows this time-reversal is the restriction of the control $U_t$ to feedback control laws  of the form~\eqref{eq:U-feedback}. This assumption ensures that the control has a Markov property, i.e. it depends only on the current value of the state, conditionally independent of its past and future. As a result, the control is adaptable to both   forward and backward filtrations generated by the state process. This  assumption is appropriate for stochastic optimal control problems since the  optimal control law is of this form.  
	\end{remark}
Before presenting the proof of this proposition, we present 
	the following Lemma that gives the expression for the Wasserstein derivative of the Hamiltonian. 
	\begin{lemma}\label{lem:W-derivative}
		If $\cev Y_t = \phi(t,\cev X_t)$, then 
		\begin{align}\nonumber
			{	\mathbb E}&\left[\frac{\partial \cev H}{\partial p}(\cev X_t,p_{t}, \cev U_t,\cev Y_t,\cev Z_t)(x)\right]\!
			=\\& 
			\frac{\partial}{\partial x} (D(x)  \frac{\partial \phi}{\partial x} (t,x)- b(x,p_t)^\top \phi(t,x)).\label{eq:W-diff-identitity}
		\end{align}
	where $p_t$ is the probability density function of ${\cev X_t}$. 
	\end{lemma}
\medskip 

	\begin{proof}
		By the definition of the Hamiltonian
		\begin{align*}
			{	\mathbb E}&\left[\frac{\partial \cev H}{\partial p}(\cev X_t,p_t, \cev U_t,\cev Y_t,\cev Z_t)(x)\right] =  {	\mathbb E}\left[\frac{\partial b}{\partial p}(\cev X_t,p_t)(x)^\top\cev Y_t\right]
			\\&=\! \int \frac{\partial b}{\partial p}(x',p_t)(x)^\top\phi(t,x')p_t(x')\ud x'. 
		\end{align*}
		Define $\beta(p):=\int b(x',p)^\top\phi(t,x')p(x')\ud x'$. Then, the identity~\eqref{eq:W-diff-identitity} implies
		\begin{align*}
			\int \frac{\partial b}{\partial p}(x',p_t)(x)^\top&\phi(t,x')p_t(x')\ud x'  =\\&  \frac{\partial \beta}{\partial p}(x) - \frac{\partial}{\partial x}(b(x,p_t)^\top\phi(t,x)). 
		\end{align*} 
		Direct evaluation of $\frac{\partial \beta}{\partial p}$, using the definition of $b(x,p)$, gives 
		\begin{align*}
			\frac{\partial \beta}{\partial p} &= -  \frac{\partial }{\partial p} \int \nabla \cdot(p(x')D(x')) \phi(t,x')\ud x' \\&=\frac{\partial }{\partial p} \int D(x') \frac{\partial\phi }{\partial x} (t,x') p(x')\ud x' = \frac{\partial}{\partial x} (D(x) \frac{\partial \phi}{\partial x} (t,x)),
		\end{align*}
		concluding the claimed identity\eqref{eq:W-diff-identitity}. 
		%
	\end{proof}

\medskip
	\newP{Proof of Prop.~\ref{prop:FBSDE-reversed}}
		Application of the It\^o rule to $\cev Y_t=\phi(t,\cev X_t)$ implies 
		\begin{align*}
			\ud \cev Y_t = &	\frac{\partial \phi}{\partial t} (t,\cev X_t)\ud t + \frac{\partial \phi}{\partial  x}(t,\cev X_t)  \frac{\partial \cev H}{\partial y}(\cev X_t,\cev U_t,\cev Y_t,\cev Z_t) \ud t\\&+ \frac{\partial \phi}{\partial  x}(t,\cev X_t)  \sigma(\cev X_t) \cev \ud W_t  - \frac{1}{2}\trace(D(\cev X_t) \frac{\partial^2 \phi}{\partial  x^2}(t,\cev X_t) ) \ud t,
		\end{align*}
		where the It\^o correction term appears with the negative sign due to the backward stochastic integration, and we used $p_t$ to denote $p_{\cev X_t}$. Using $\cev Z_t = \frac{\partial \phi}{\partial x} (t,\cev X_t)\sigma(\cev X_t)$, $ \frac{\partial \cev H}{\partial y} = \frac{\partial H}{\partial y} + b$, and the fact that $\phi$ solves~\eqref{eq:FBSDE-PDE},  
		\begin{align*}
			\ud \cev Y_t = & \frac{\partial \phi}{\partial  x}(t,\cev X_t) b(\cev X_t,p_t)\ud t - \frac{\partial H}{\partial x}(\cev X_t,\cev U_t, \cev Y_t,\cev Z_t) \ud t\\&+\cev Z_t \cev \ud W_t  - \trace(D(\cev X_t) \frac{\partial^2 \phi}{\partial  x^2}(t,\cev X_t) ) \ud t. 
		\end{align*}
		On the other-hand, starting from~\eqref{eq:backward-sde-reversed}, using the definition~\eqref{eq:super-H},  and the identity~\eqref{eq:W-diff-identitity}  
		\begin{align*}
			\ud \cev Y_t = &- \frac{\partial H}{\partial x}(\cev X_t,\cev U_t, \cev Y_t,\cev Z_t) \ud t - \frac{\partial b}{\partial x}(\cev X_t,p_t)\cev Y_t\ud t \\&- 2\frac{\partial}{\partial x} \trace(\sigma(\cev X_t)^\top \cev Z_t)\ud t 	 -  \frac{\partial}{\partial x} (D(\cev X_t)  \frac{\partial \phi}{\partial x} (t,\cev X_t))\\&+ \frac{\partial}{\partial x}  (b(\cev X_t,p_t)\phi(t,\cev X_t)) + \cev Z_t \cev \ud W_t.  
		\end{align*}
		Comparing the two derivations and using the  following  two identities 
		\begin{align*}
			&\frac{\partial}{\partial x}  (b(x,p)\phi(t,x)) =  \frac{\partial \phi}{\partial  x}(t,x) b( x,p) + \frac{\partial b}{\partial x}(x,p)\phi(t,x)\\
			& \frac{\partial}{\partial x} (D(x)  \frac{\partial \phi}{\partial x} (t,x)) =  \trace(\frac{\partial D}{\partial x}(x)\frac{\partial \phi}{\partial x}(t,x)) + \trace(D(x) \frac{\partial^2 \phi}{\partial  x^2}(t,x) ). 
		\end{align*}
	concludes~\eqref{eq:FBSDE-reversed-1}. 
Finally, the equality~\eqref{eq:FBSDE-reversed-2} follows by application of Thm.~\ref{thm:sde-reversal} with $\tilde a(t,x) = a(x,k(t,x,\phi(t,x)),\frac{\partial \phi}{\partial x}(t,x) \sigma(x) )$. \hfill \QED
	\section{Linear quadratic setting}\label{sec:LQ}
	In this section, we consider the special case where  the dynamic model is linear, the initial distribution is Gaussian, and the control cost is quadratic, i.e. 
	\begin{align*}
		&a(x,u) = Ax + Bu,\quad \sigma(x)=\sigma,\quad p_0 = \mathcal N(m_0,\Sigma_0)\\
		&\ell(x,u) =\frac{1}{2}\left[x^\top Q x +u^\top Ru \right],\quad \ell_f(x) = \frac{1}{2}x^\top Q_f x 
	\end{align*}
	where 
	%
	%
 $Q,Q_f\succeq$ are assumed to be positive semi-definite, 
	$R\succ 0$  is assumed  to be positive definite, and $\mathcal N(m,\Sigma)$ denotes a Gaussian distribution with mean $m$ and covariance $\Sigma$.   
%
Under the LQ setting, the Hamiltonian~\eqref{eq:H}  becomes 
\begin{equation*}
	H(x,u,y,z) := \frac{1}{2}\left[x^\top Q x +u^\top R u \right]+ y^\top (Ax+Bu) + \trace(\sigma^\top z),
\end{equation*}
concluding the  FBSDE
%
%
%
%
\begin{subequations}
	\begin{align}\label{eq:forward-sde-lin}
		\ud X_t &= AX_t  \ud t+ B U_t \ud t +\ud  W_t,\quad X_0\sim p_0\\\label{eq:backward-sde-lin}
		-\ud Y_t &= A^\top Y_t \ud t + Q X_t \ud t - Z_t \ud W_t,\quad Y_{t_f} = Q_f X_{t_f}. 
	\end{align}
	Moreover, the control input that minimizes the Hamiltonian, according to~\eqref{eq:min-condition}, is equal to 
	$U_t = -R^{-1}B^\top Y_t$. 
	
	We present the solution of the FBSDE system above for an affine form of the control according to 
	\begin{align}\label{eq:affine-U}
	U_t = K_1(t)X_t+ K_2(t)Y_t + K_3(t) Z_t + K_4(t)
	\end{align}
\end{subequations}
for some time-varying parameters $K_1$, $K_2$,$K_3$, and $K_4$ with appropriate dimensions.  
	\begin{proposition}
		Consider the FBSDE~\eqref{eq:forward-sde-lin}-\eqref{eq:backward-sde-lin} 
		with affine form of control as in~\eqref{eq:affine-U}.  Then, the FBSDE is solved with $Y_t =  G_1(t)X_t + G_2(t)$ and $Z_t = G_1(t)\sigma$ where $G_1$ and $G_2$ solve matrix differential equations 
		\begin{subequations}
		\begin{align}\nonumber
			&\dot G_1(t) + G_1(t)A +A^\top G_1(t) +G_1(t)B K_1(t) \\&+  G_1(t)B K_2(t)G_1(t) +Q = 0,\quad G_1(T)=Q_f,\label{eq:G1}\\
			&\dot G_2(t) + G_1(t) B (K_2(t) G_2(t) \nonumber + K_3(t) G_1(t)\sigma + K_4(t))\\& + A^\top G_2(t) = 0,\quad G_2(T)=0. \label{eq:G2}
		\end{align}  
		For the case with the optimal control law, i.e. $K_2(t) = - R^{-1}B^\top$ and $K_1(t)=K_3(t)=K_4(t)=0$, we have $G_2(t)=0$ and $G_1(t)$ solving the Riccati equation
	\begin{align}\nonumber 
		\dot G_1(t)&+ G_1(t) A + A^\top G_1(t) + Q \\&- G_1(t) BR^{-1}B^\top G_1(t),\quad G_{1}(T) = Q_f. \label{eq:Ricatti}
	\end{align}
\end{subequations}
	\end{proposition}
	\medskip 
	\begin{proof}
		The proof follows upon application of the It\^o rule to $Y(t)=G_1(t)X_t  + G_2(t)$ and using the formula for $\dot{G}_1(t)$ and  $\dot{G}_2(t)$. 
	\end{proof}

\subsection{Time reversal in LQ setting} 
We like to express the time-reversed FBSDE~\eqref{eq:forward-sde-reversed}-\eqref{eq:backward-sde-reversed} in the LQ setting. We assume the control input $\cev U_t$ has the affine form~\eqref{eq:affine-U}. Then, according to Prop.~\ref{prop:FBSDE-reversed}, the probability law of $(\cev X_t,\cev Y_t,\cev X_t,\cev U_t)$ is equal to $(X_t,Y_t,X_t, U_t)$. 
As a result, the probability law of $\cev X_t$, denoted by $p_t$, is  Gaussian, similar to $X_t$, and the drift term  $b(x,p_t) =D\Sigma_t^{-1}(x-m_t) $,  where $m_t$ and $\Sigma_t$ denote the mean and covariance of $p_t$, respectively. This conlcudes the following SDE for $\cev X_t$ 
%
%
\begin{subequations}
\begin{align}\label{eq:forward-sde-lin-reversed}
		\ud \cev X_t = A\cev  X_t \ud t+ B\cev U_t \ud t +  \sigma \cev \ud {W}_t+  &D\Sigma_t^{-1}(\cev X_t - m_t)\ud t,
\end{align}
where $D := \sigma \sigma^\top$. In order express the adjoint process, we evaluate  the  Wasserstein derivative of the Hamiltonian according to Lemma~\ref{lem:W-derivative}, using $\cev Y_t = G_1(t) \cev X_t + G_2(t)$, 
 	\begin{align*}
 	{	\mathbb E}&\left[\frac{\partial \cev H}{\partial p}(\cev X_t,p_t, \cev U_t,\cev Y_t,\cev Z_t)(x)\right]\!
 	=  -D\Sigma_t^{-1} (G_1(t)x+G_2(t))\\&\quad - G_1(t)  D\Sigma_t^{-1}(x-m_t),  
 \end{align*}
concluding the following expression for the adjoint process: 
	\begin{align}\label{eq:backward-sde-lin-reversed}
	-\ud \cev Y_t &= A^\top \cev Y_t \ud t + Q  \cev X_t \ud t - \cev Z_t \ud \cev W_t  \\&-G_1(t)  D\Sigma_t^{-1}(\cev X_t-m_t)  \ud t, 	\quad \cev Y_{t_f} = Q_f \cev X_{t_f}. \nonumber
\end{align}
Application of the It\^o rule verifies that ~\eqref{eq:backward-sde-lin-reversed} is indeed solved with  $\cev Y_t = G_1(t) \cev X_t + G_2(t)$ and $\cev Z_t = G_1(t) \sigma $ where $G_1$ and $G_2$ solve~\eqref{eq:G1}-\eqref{eq:G2}. In the proposed numerical scheme, we do not solve~\eqref{eq:G1}, but compute  the matrix $G_1(t)$ by solving the regression problem
\begin{equation*}
	G_1(t) = \argmin_G\,\mathbb E\left[ \| \cev Y_t - G \cev X_t\|^2 \right] = \text{Cov}(\cev X_t,\cev Y_t)\text{Cov}(\cev Y_t,\cev Y_t)^{-1}
\end{equation*}
where $\text{Cov}(V_1,V_2)$ denotes the covariance matrix of two random variables $V_1$ and $V_2$. 
\end{subequations}

\begin{figure*}
	\centering
	\includegraphics[width=0.3\hsize,trim={50pt 0pt 50pt 0 pt},clip]{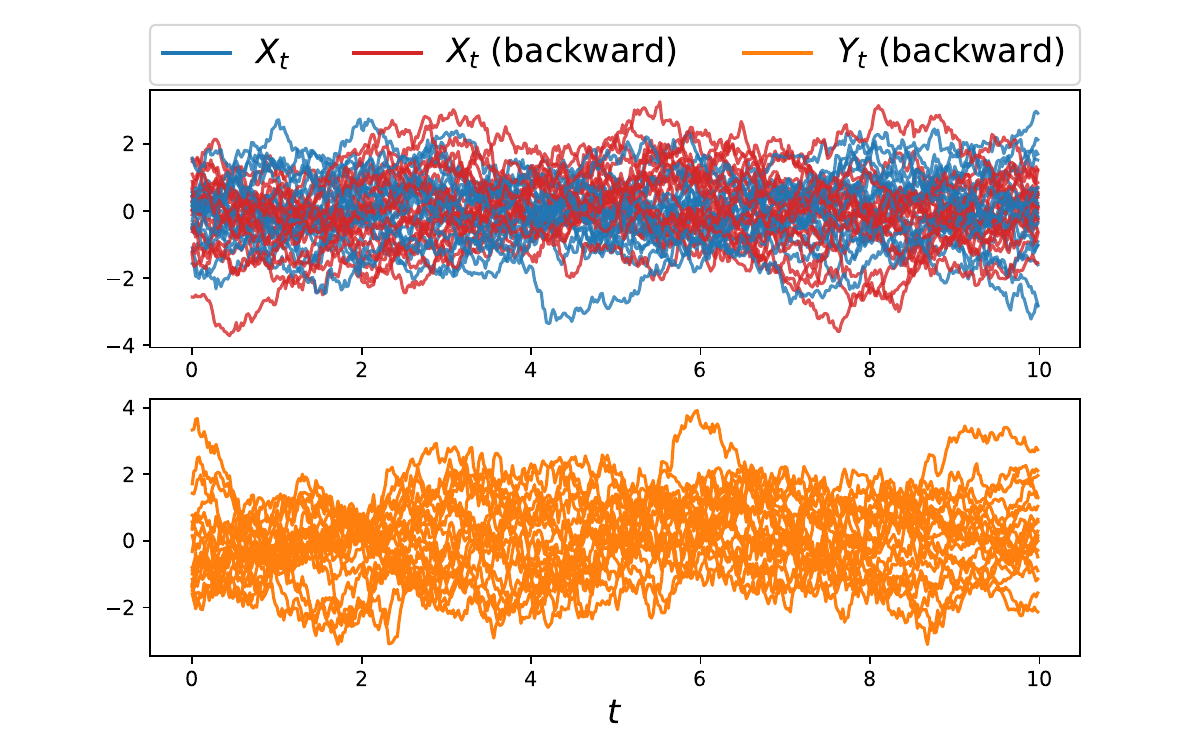} \hfill 
	\includegraphics[width=0.3\hsize,trim={50pt 0pt 50pt 0 pt},clip]{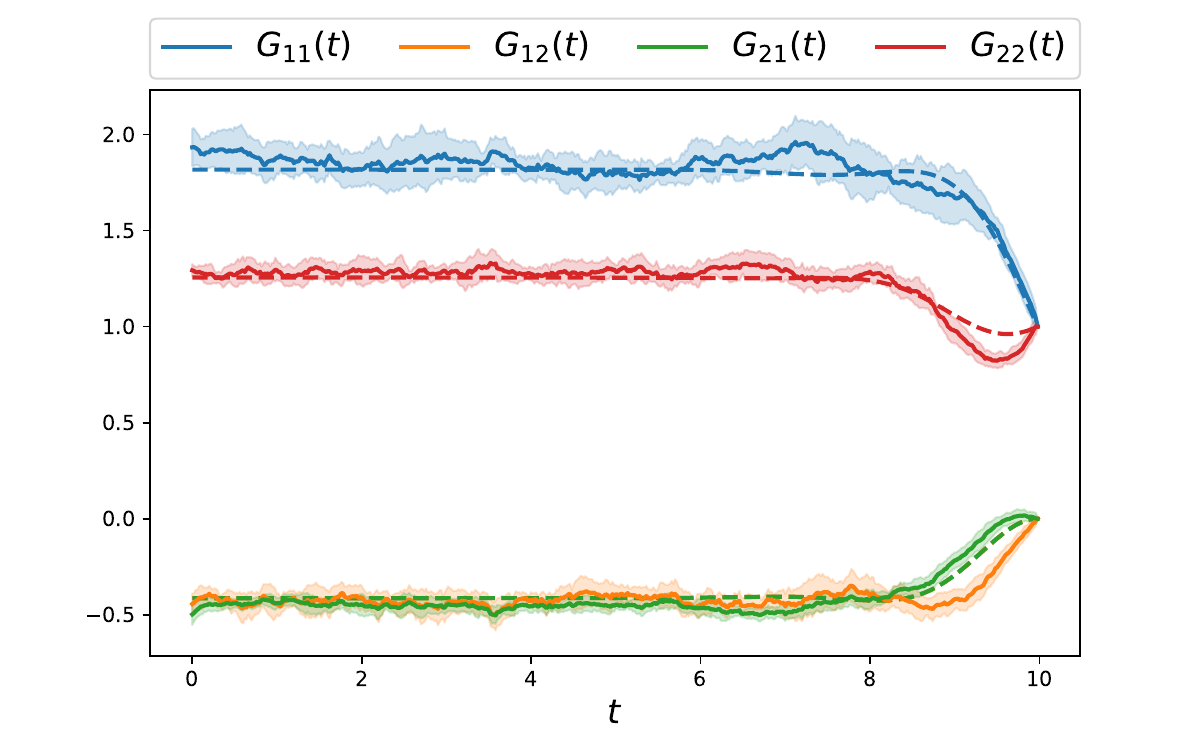} \hfill 
	\includegraphics[width=0.3\hsize,trim={50pt 0pt 50pt 0 pt},clip]{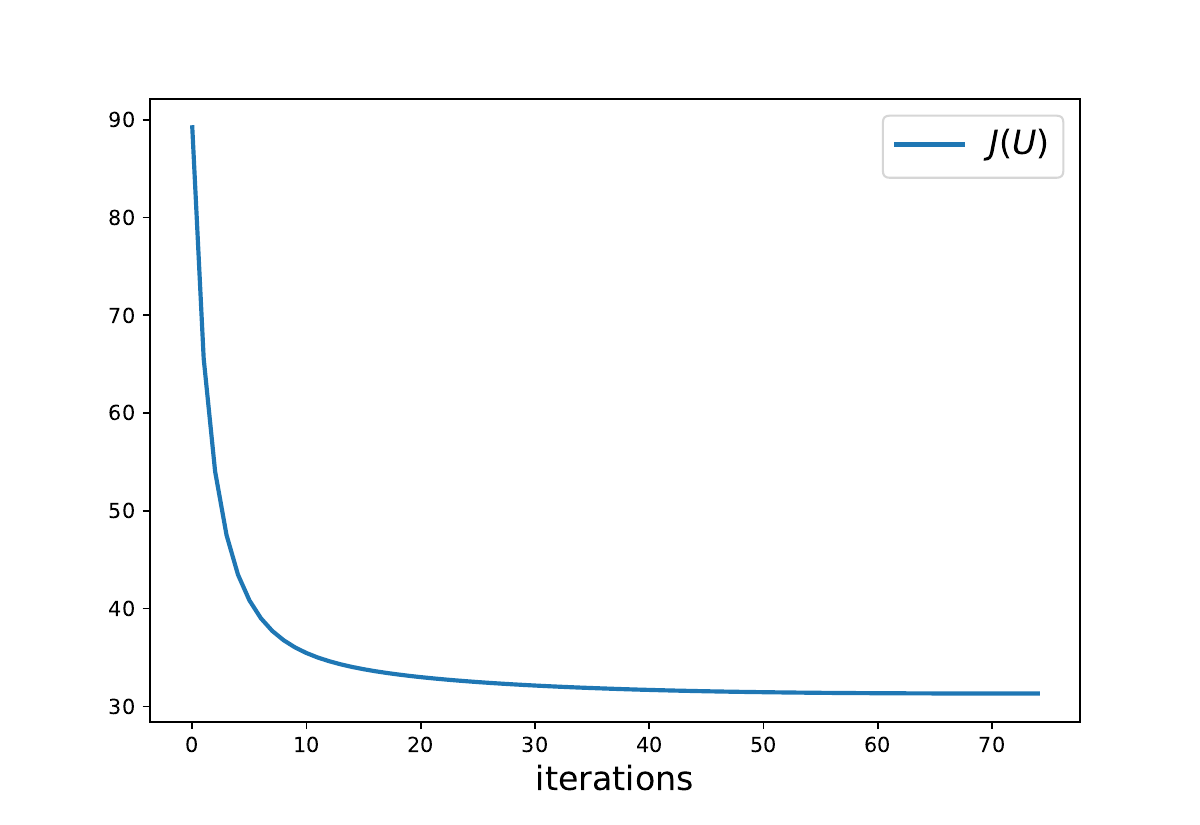}\\
	(a) \hspace{160pt} (b)  \hspace{160pt}  (c) 
	\caption{Numerical results for a two-dimensional linear quadratic stochastic optimal control problem presented in Sec.~\ref{sec:num-example}: (a) Sampled trajectories of the state~\eqref{eq:alg-forward}, its time-reversal~\eqref{eq:alg-reverse}, and the adjoint process~\eqref{eq:alg-adjoint}; (b) numerical approxiamtion of the matrix $G(t)$ according to~\eqref{eq:regression} in comparison to its exact value according to~\eqref{eq:Ricatti}; (c) value of the objective function~\eqref{eq:control-cost} as the number of algorithm iterations increases.}
	\label{fig:numerics}
\end{figure*}

\section{Numerical Algorithm} \label{sec:numerics}
A numerical procedure  is proposed to approximate the 
%
%
solution to the FBSDE~\eqref{eq:forward-sde-lin}-\eqref{eq:backward-sde-lin}, its time-reversal~\eqref{eq:forward-sde-lin-reversed}-\eqref{eq:backward-sde-lin-reversed}, and the optimal control input. While the procedure is presented for the LQ setting, it is generalizable to the nonlinear case as explained in Remark~\ref{remark:nonlinear}.

\begin{subequations}
\newP{Initialization}
The algorithm begins with generating $N$ independent samples $\{X^i_0\}_{i=1}^N$ from the initial distribution $p_0=\mathcal N(m_0,\Sigma_0)$.  Each sample $X^i_0$, for $i=1,\ldots,N$,  is then paired with  a zero  initialization of the  control inputs $U^{i,0}_t=\cev U^{i,0}_t=0$, and a random realization of the Weiner processes $W^i_t$ and $\tilde W^i_t$, that will be used for the simulation of the FBSDE and its time-reversal, respectively. 

\newP{Iterations}
The $k$-th iteration, for $k=1,2,3,\ldots,k_ f$, involves the following steps:
\begin{enumerate}
	\item Forward simulations of~\eqref{eq:forward-sde-lin} according to 
	\begin{align}\label{eq:alg-forward}
		\ud X^{i,k}_t \!= \!AX^{i,k}_t\ud t + B  U^{i,k-1}_t\ud t + \sigma \ud W^i_t,~ X^{i,k}_0=X^i_0
	\end{align}
for $i=1,\ldots,N$ and $t\in[0,T]$. 
\item Approximation of the  F\"olmer's drift. For the LQ setting, this amounts to approximating the mean  and covariance according to
\begin{align*}
	m^{k}_t \approx \frac{1}{N}\sum_{i=1}^NX^{i,k}_t,~ \Sigma^{k}_t \approx \frac{1}{N}\sum_{i=1}^N (X^{i,k}_t - m^{k}_t)(X^{i,k}_t - m^{k}_t )^\top. 
\end{align*}
	\item Simulation of time-reversed FBSDE~\eqref{eq:forward-sde-lin-reversed}-\eqref{eq:backward-sde-lin-reversed} according to 
	\begin{align}\label{eq:alg-reverse}
		&\ud  \cev X^{k,i}_t = A \cev X^{k,i}_t \ud t+ B \cev U^{i,k-1}_t \ud t + \sigma \cev \ud {\tilde W}^i_t\\&+ D (\Sigma^{k}_{t})^{-1}(\cev X^{i,k}_t - m^{k}_{t})\ud t,
		~   X^{i,k}_{T} \sim \mathcal N(m^{k}_{T},\Sigma^{k}_{T}),\nonumber\\\label{eq:alg-adjoint}
		-&\ud \cev Y^{i,k}_t = A^\top  \cev Y^{i,k}_t \ud t + Q  \cev X^{i,k}_t \ud t - G_1^{k}(t) \sigma \cev  \ud \tilde W^i_t \\& - G_1^{k}(t)D (\Sigma^{k}_{t})^{-1}(\cev X^{i,k}_t - m^{k}_{t})\ud t , 	\quad   Y^{i,k}_{T} = Q_f X^{i,k}_{T},\nonumber
	\end{align}
for $i=1,\ldots,N$ and $t\in[0,T]$, 
	where 
	\begin{align}\label{eq:regression}
		G_1^{k}(t) &= \argmin_{G}\, \frac{1}{N}\sum_{i=1}^N \|\cev Y^{i,k}_t - G\cev X^{i,k}_t \|^2. 
	\end{align}
	\item Gradient descent update of the control to minimize the Hamiltonian according to 
	\begin{align*}
		 U^{i,k}_t &=  U^{i,k-1}_t - \eta \frac{\partial H}{\partial u} (X^{i,k}_t, U^{i,k-1}_t,Y^{i,k}_t,Z^{k}_t),\\
		\cev U^{i,k}_t &= \cev U^{i,k-1}_t - \eta \frac{\partial H}{\partial u} (\cev X^{i,k}_t, \cev U^{i,k-1}_t,\cev Y^{i,k}_t,\cev Z^{k}_t),
	\end{align*}
	for a small step-size $\eta>0$, where $Z^{k}_t = 	\cev Z^{k}_t = G_1^{k}(t)  \sigma$, $Y^{i,k}_t = G_1^{k}(t)X^{i,k}_t$, and $\frac{\partial H}{\partial u} (x,u,y,z)  = Ru + B^\top y$. 
\end{enumerate}
\medskip 

\end{subequations}


\newP{Output} At the end of the iterations, the samples  $(X^{i,k_f}_t,U^{i,k_f}_t,G^{k_f}_1(t)X^{i,k_f}_t,G^{k_f}_1(t)\sigma)_{i=1}^N$ 
  are expected to approximate the probability law of the quadruple $(\hat X, \hat U, \hat Y, \hat Z)$ that satisfy the SMP conditions~\eqref{eq:forward-sde}-\eqref{eq:backward-sde}-\eqref{eq:min-condition} in the LQ setting. 
A feedback control law is obtained according to  $U_t=\argmin_u H(X_t,u,Y_t,Z_t) = -R^{-1}B^\top Y_t\approx -R^{-1}B^\top G^{k_f}_1(t)X_t$, given a realization of the state process $X_t$.  

\medskip 
\begin{remark}\label{remark:nonlinear}
	In order to extend  the procedure to the nonlinear setting, it is necessary to use a score function approximation algorithm in step 2, and solve a nonlinear regression problem 
$
	\min_{\phi}\,\frac{1}{N}\sum_{i=1}^N \|\cev Y^{i,k}_t - \phi(\cev X^{i,k}_t) \|^2
$
to find $\phi^k(t,x)$ 
for all $t\in [0,T]$ in step 3. When $\phi^k(t,x)$  is available, we can evaluate the Wasserstein derivative according to Lemma~\ref{lem:W-derivative} and  use $\cev Z_t= \frac{\partial \phi}{\partial x} (t,\cev X_t) \sigma(\cev X_t)$ to find $\cev Z_t$. 
\end{remark}

\subsection{Numerical Example} \label{sec:num-example}
The algorithm is illustrated for a two-dimensional mass-spring example with the following model parameters 
\begin{align*}
&A = \begin{bmatrix}
	0 & 1 \\
	-1 & 0
\end{bmatrix},~B = \begin{bmatrix}
0\\1
\end{bmatrix},~\sigma = \begin{bmatrix}
1 & 0 \\
0 & 1
\end{bmatrix},\\&R=1,~m_0=\begin{bmatrix}
0\\0
\end{bmatrix},\quad Q=Q_f=\Sigma_0=\begin{bmatrix}
1 & 0 \\
0 & 1
\end{bmatrix}.
\end{align*}
The SDEs are simulated using the Euler-Maruyama method  with  $\Delta t = 0.02$.  
The number of samples  $N=1000$, iterations $k_f=75$, and optimization step-size $\eta=0.02$.  The numerical results are shown in Fig.~\ref{fig:numerics}. Panel~(a) presents the first component of the trajectories $X^{i,k_f}_t$, $\cev X^{i,k_f}_t$, and $\cev Y^{i,k_f}_t$ for $20$ randomly selected samples.  Panel~(b)  shows the components of the matrix $G^{k_f}(t)$, averaged over $10$ independent simulations, in comparison with their exact value, derived by solving the Riccati equation~\eqref{eq:Ricatti}. Finally, panel~(c) shows the convergence of the value of the objective function~\eqref{eq:control-cost}  as the number of iterations increases. 
\section{Conclusion}
This paper presents the time-reversal of the FBSDE that appears in stochastic maximum principle. The time-reversal is used to propose an iterative numerical algorithm to approximately solve the FBSDE and obtain the optimal control input. The algorithm is presented and numerically illustrated  for the LQ setting. The extension to the nonlinear setting and subsequent comparison to the related algorithms is the subject of future research.   

\bibliographystyle{IEEEtran}
\bibliography{references}

\end{document}